\documentclass[a4paper,12pt]{article}
\usepackage{graphicx}
\usepackage{multirow}
\usepackage[labelfont=bf,font=normalsize]{caption}
\usepackage{float}
\usepackage{footnote}
\usepackage{amssymb}
\usepackage{pifont}
\usepackage{ifpdf}
\ifpdf
\else
\usepackage{pstcol,pst-fill,pstricks}
\fi
\usepackage{natbib}
\usepackage{mathpazo}
\usepackage{import}
\usepackage{enumerate}
\usepackage[latin1]{inputenc}
\usepackage[OT1]{fontenc}
\usepackage[T1]{fontenc}

\usepackage[margin=1in]{geometry}
\usepackage[doublespacing]{setspace}

\usepackage{color}
\usepackage{mdframed}
\usepackage{tikz}
\usepackage[bottom,flushmargin]{footmisc}
\usepackage{subcaption}  
\usepackage{times}
\usepackage{fancyhdr,graphicx,amsmath,amssymb}
\usepackage{mathtools}
\usepackage{amsmath}
\usepackage{amsthm}
\usepackage{thmtools}
\usepackage{booktabs}

\makeatletter
\newcommand*\bigcdot{\mathpalette\bigcdot@{.5}}
\newcommand*\bigcdot@[2]{\mathbin{\vcenter{\hbox{\scalebox{#2}{$\m@th#1\bullet$}}}}}
\makeatother
\usepackage{authblk}
\usepackage{setspace}
 \newcommand{\be}{\begin{equation}}
\newcommand{\ee}{\end{equation}}

\newcommand{\cbar}{\overline{c}}
\newcommand{\kbar}{\overline{k}}
\def\dotk{\dot{k}}
\def\dotc{\dot{c}}
\def\cbar{\bar{c}}
\def\kbar{\bar{k}}
\newcommand{\kmin}{k_{\min}}
\newcommand{\kmax}{k_{\max}}

\DeclareMathAlphabet{\pazocal}{OMS}{zplm}{m}{n}
\allowdisplaybreaks

\usepackage[hyperfootnotes=false]{hyperref}
 \hypersetup{colorlinks=true,
 linkcolor=blue,
 citecolor=blue,
 filecolor=black,
 urlcolor=black,
 pdfstartview={Fit},
pdfpagemode=UseNone
 }
\usepackage{cleveref}

\declaretheoremstyle[
  headfont=\normalfont\scshape,
  numbered=unless unique,
  bodyfont=\normalfont,
  spaceabove=1em,
  spacebelow=1em,
]{exmpstyle}

\newcommand{\p}{\partial}

\newcommand{\defeq}{\vcentcolon=}

\def\ee{\mathsf{e}}

\usepackage[mathlines]{lineno}

\newtheorem{Ass}{Assumption}
\newtheorem{prop}{Proposition}
\newtheorem{lemma}{Lemma}
\newtheorem{remark}{Remark}
\theoremstyle{definition}

\newtheorem{theorem}{Theorem}
\newtheorem{corollary}{Corollary}

\title{An Arbitrarily Precise Global Closed Form Approximation for the Neoclassical Growth Model\thanks{I thank two anonymous referees for their insightful suggestions.}}
\author[1]{Jordan Roulleau-Pasdeloup}
\affil[1]{Banque de France}
\date{\today}

\begin{document}
\begin{titlepage}
\maketitle
\thispagestyle{empty}

\abstract{I consider a neoclassical growth model with a constant absolute risk aversion (CARA) utility function and derive a global closed form approximation that is arbitrarily precise as the discount rate $\rho$ is close to the population growth rate $n$. I use it to show that the consumption function is strictly concave and that countries can have two different paths converging to the steady-state: front-loading and back-loading.}
\vspace{.5cm}
\end{titlepage}
\setstretch{1.5}
\setlength{\parskip}{1em}

\section{Introduction}

The \cite{Ramsey1928mathematical}-\cite{Cass1965optimum}-\cite{Koopmans1963concept} model is a precursor of both modern growth theory as well as business cycle theory. As a result, it has been the subject of considerable attention and is being taught in virtually every graduate program in Economics. Frustratingly, there is no known closed-form solution that applies for standard preferences and technology. 

Given that this is a model of growth, we want to know how an economy transitions from a low stock of initial capital to its steady state. In the absence of a closed-form solution, this question is usually answered with a phase diagram\textemdash see \cite{Barro2004economic} for a seminal textbook treatment. This however only offers indirect insights about how the model works and it prevents one from understanding how a specific parameter shapes the growth trajectory. As a result, most of the explicit results that we have regarding this class of models are due to a linear approximation around the steady state. By construction then, this approach isn't well suited to study how the economy behaves when starting with a low level of capital.

Against this backdrop, many economists have been searching for special cases where a closed-form solution can be obtained. \cite{LongPlosser1983} solve a stochastic Real Business Cycle model with full depreciation and logarithmic utility.\footnote{See \cite{StokeyLucas1989} or more recently \cite{Heer2024} for a textbook treatment of the deterministic version of this model.} Models in which the production function is linear in capital (so-called "AK" models) also admit a closed form solution\textemdash see \cite{Barro2004economic} for an explicit derivation. Models in which the saving rate is constant along the optimal path, as in \cite{Kurz1966general}, also admit a closed form solution.\footnote{\cite{Kurz1966general} studies under which parameter configuration of the \cite{Ramsey1928mathematical} model a constant saving rate is \emph{optimal}, and the content of the result is that the model then reduces to \cite{Solow1956contribution}, whose accumulation equation is a Bernoulli equation under Cobb-Douglas.}. In a growth model where saving is endogenous, \cite{Barro1995capital} show that if the elasticity of intertemporal substitution is tied to the rates of discounting and depreciation in a specific way, then the savings rate is constant. Likewise, if capital's share in GDP is equal to the inverse of the intertemporal elasticity of substitution as in \cite{Smith2006closed}, the ratio of consumption to capital is constant and the model boils down to a Bernoulli differential equation. The common theme in these contributions is that saving is a linear function of output, which makes the model tractable.\footnote{An exception is \cite{Mehlum2005closed}, who derives an explicit non-linear saddle path when the production function is Leontieff.}

This begs for the following question: are there more situations where one can find a closed form solution? Given the results in \cite{Gun2021ramsey} and \cite{Koster2025closed}, for constant relative risk aversion (CRRA) utility functions and Cobb-Douglas production functions the answer is most likely \textit{no}.\footnote{The two contributions differ: \cite{Gun2021ramsey} use Lie-group methods and recover, in their integrability condition, the \cite{Kurz1966general} solution, whereas \cite{Koster2025closed} show that the three known closed-form cases exhaust the configurations that are integrable in elementary functions.} In this context, I show in this paper that one can make some progress by considering a constant \textit{absolute} risk aversion (CARA) utility function as in the seminal contributions of \cite{Merton1971optimum}, \cite{Kimball1989precautionary} and \cite{Caballero1990consumption}. The assumption of CARA utility is no silver bullet and I show that for a generic parameterization the system of differential equations boils down to an Abel equation of the second kind for which no closed form solutions are known. However, if the discount rate $\rho$ is equal to the population growth rate $n$, this system admits a closed form solution in the Lambert W function as in \cite{Roulleau2026explicit}. This case is slightly problematic because the generalized rate of time preference $\rho-n$ is then 0 and the usual transversality condition doesn't hold as a result. For this reason, I show explicitly that this case provides an arbitrarily precise solution for the case where $\rho\to n^+$.

It is worth being precise about the sense of "closed form" used here: what the current paper delivers under the assumption that $\rho=n$ is an explicit expression for the policy functions for savings $s(k)$ and consumption $c(k)$ in terms of the Lambert $W$ function\textemdash a function defined by the inverse of $h(w)=we^w$, in a way that is similar to the error function or the hypergeometric functions used by \cite{Boucekkine2008special} for the Lucas-Uzawa model. The Lambert W function is not Liouvillian in the sense that it cannot be written as a finite number of algebraic operations \textemdash see \cite{Bronstein2008algebraic} for a recent proof. That being said, there have been calls (see \cite{Gouvea2000time}) to recognize the Lambert W function as an elementary function.

With these explicit non-linear equations for optimal consumption/saving in hand, I focus on the dynamics of this economy starting with a low level of capital. In other words, the goal is to study how fast the initial take-off is on the way to the steady state. I show that the consumption function is strictly concave, so that consumption growth is fast at first and slows down as the economy converges to its steady state. In addition, I show that there are two possible paths in terms of the speed of convergence of the stock of capital to its steady state. With relatively low risk aversion, the households save a lot early on so that capital adjusts faster starting from a low level compared to when it is closed to its steady state\textemdash front-loading. With relatively high risk aversion, the reverse happens so that the speed of adjustment is slow initially and faster around the steady-state \textemdash back-loading.

\section{A Neoclassical Growth Model with CARA Utility}
\label{sec:model}

Time is continuous and population evolves exogenously according to $L(t) = L_{0} e^{n t}$ with $n
\geq 0$. Production per worker is given by a Cobb-Douglas technology $f(k) = k^{\alpha}$
with $\alpha \in (0, 1)$, where $k$ denotes capital per worker. Whenever no ambiguity arises I drop the time subscript to economize on notation. Capital depreciates at
rate $\delta > 0$ and the representative household has CARA preferences defined relative to a subsistence level $\bar{c} \geq 0$:
\begin{equation} \label{eq:utility}
u(c) = -\frac{1}{\gamma} e^{-\gamma(c - \bar{c})}, \qquad \gamma > 0,
\end{equation}
defined on the domain $c \geq \bar{c}$. The parameter $\bar{c}$ represents the
subsistence consumption level: consumption below $\bar{c}$ is infeasible, and utility
approaches its lower bound as $c \rightarrow \bar{c}^+$. The coefficient of absolute risk
aversion $-u''(c)/u'(c) = \gamma$ is constant and independent of $\bar{c}$. The main reason for including this is to prevent consumption from being negative, which is a possibility with a CARA utility function \textemdash see \cite{Blanchard1988consumption}. The
household discounts future utility at the rate of time preference $\rho > 0$.
The net rate of return on capital is:
\begin{equation}
r(k) := f'(k) - \delta = \alpha k^{\alpha - 1} - \delta.
\label{eq:rk_def}
\end{equation}
Technology is assumed to be constant and normalized to one. Given that the utility function is not homothetic, this economy has no balanced growth path \footnote{See \cite{King1988production}. The classic statement of the stylized facts that these preferences can match is studied in \cite{Kaldor1961capital}, and \cite{Kongsamut2001beyond} and \cite{Ngai2007structural} study what happens when they are relaxed.}

\subsection{The Planner's Problem}

The social planner solves the following maximization program:
\begin{align}
\label{eq:disc_util}
\max_{c(t)}&\quad L_{0} \int_{0}^{\infty} e^{-(\rho - n) t} u(c(t)) dt\\
\label{eq:resource_constraint}
\dotk(t) &= k(t)^\alpha - (\delta + n) k(t) - c(t)\\
\label{eq:subsistence_constraint}
c(t) &\geq \bar{c}
\end{align}
where $L_{0}$ is the initial population and the initial stock of capital is given by $k(0)=k_0$. The effective rate of discount is $\rho - n$. I maintain throughout the parametric restriction $\rho > n$ which ensures that the standard transversality condition holds. When the subsistence constraint does not bind, the first-order conditions reduce to the standard Ramsey system with the CARA Euler equation,
\begin{align}
\dotk = k^{\alpha}-(\delta+n)k-c, \qquad \dotc = \frac{r(k)-\rho}{\gamma},
\label{eq:ramsey_explicit}
\end{align}
which is essentially the textbook treatment. The constrained regime described below is a departure from this benchmark. The current-value Hamiltonian for the planner's problem is:
\begin{equation} \label{eq:hamiltonian}
\mathcal{H}(k, c, \mu, \lambda) = u(c) + \mu \bigl[ k^\alpha - (\delta + n) k - c
\bigr] + \lambda (c - \bar{c}),
\end{equation}
where $\mu$ is the current-value costate variable (shadow price of capital) and
$\lambda \geq 0$ is the Kuhn-Tucker multiplier on the subsistence constraint $c \geq
\bar{c}$. The necessary conditions for an interior maximum of $\mathcal{H}$ in $c$, together with
the costate equation and complementary slackness are:
\begin{align}
& u'(c) - \mu + \lambda = 0,
\label{eq:foc} \\
& \lambda \geq 0, \quad c \geq \bar{c}, \quad
\lambda (c - \bar{c}) = 0, \label{eq:cs} \\
& \dot{\mu} = (\rho - n) \mu -
\frac{\partial \mathcal{H}}{\partial k} = \bigl[\rho - r(k) \bigr] \mu,
\label{eq:costate} \\
\label{eq:resource_constraint_2}
&\dotk = k^\alpha - (\delta + n) k - c\\
&\lim_{t \to \infty} e^{-(\rho - n) t} \mu(t) k(t)
= 0. \label{eq:tvc}
\end{align}
Before proceeding forward, it will be useful to describe the unconstrained steady state as it will imply a number of restrictions on the admissible range for initial capital as well as the value for subsistence consumption. From equation \eqref{eq:costate}, the steady state level of capital $k^\ast$ will be such that $r(k^\ast)=\rho$, which implies that steady state capital and consumption are given by:
\begin{align}
k^\ast(\rho)=\left(\frac{\alpha}{\rho+\delta}\right)^{\frac{1}{1-\alpha}},\quad c^\ast = f(k^\ast)-(\delta+n)k^\ast,   
\label{eq:steady_states}
\end{align}
Defining the net per-worker output $\psi(k):=f(k)-(\delta+n)k$, the following Lemma derives $k_{\min}$, the minimum admissible value for the initial stock of capital.

\begin{lemma}[Subsistence level for capital]\label{lem:kmin}
Assume that $\cbar<c^\ast$. Then \emph{(i)} $\psi$ is strictly concave on $(0,\infty)$, vanishes at $0$, attains its maximum at $k_{\psi}\defeq (\alpha/(\delta+n))^{1/(1-\alpha)}$, is strictly increasing on $(0,k_{\psi})$ and strictly decreasing on $(k_{\psi},\infty)$ with $\psi(k)\to -\infty$ as $k\to\infty$; \emph{(ii)} $k^\ast(\rho)< k_{\psi}$ for every $\rho> n$; \emph{(iii)} $\psi(k)=\cbar$ has exactly two roots $0<\kmin<\kmax$, and $\psi(k)-\cbar>0$ on $(\kmin,\kmax)$ and $\psi(k)-\cbar<0$ outside $[\kmin,\kmax]$ ; \emph{(iv)} $\kmin<k^\ast(\rho)<\kmax$
\end{lemma}
\begin{proof}
See Appendix \ref{sec:app_proof_Lemma_kmin}.
\end{proof}
The main purpose of Lemma \ref{lem:kmin} is to prove the existence of $k_{min}$. If initial capital is such that $k_0<\kmin$ then the constraint set of \eqref{eq:disc_util}- \eqref{eq:subsistence_constraint} is \emph{empty}: feasibility requires $c\ge\cbar$ and $k\ge0$, but under the most parsimonious policy $c=\cbar$ implies $\dotk=\psi(k)-\cbar<0$ on $[0,\kmin)$ by Lemma \ref{lem:kmin}\emph{(iii)}. This implies that capital falls, $\psi(k)-\cbar$ grows more negative, and $k$ reaches zero in finite time; any other admissible policy runs capital down at least as fast. Hence no feasible path exists, and the restriction $k_0>\kmin$ is a well-posedness condition. Accordingly, I maintain the following assumption throughout the paper.

\begin{Ass}
\label{ass:k_min_c_bar}
The initial value for the stock of capital is such that $k_{\min}<k_0<k^\ast(\rho)$ and the level of subsistence consumption is such that $\cbar<c^\ast$.
\end{Ass}

By enforcing the upper bound on $k_0$, I focus on an economy that converges to its steady state from below, which is the object of interest for the take-off. I proceed sequentially and check that the equilibrium conditions \eqref{eq:foc}-\eqref{eq:tvc} hold above and at $\bar{c}$. I begin with the interior regime where $c>\bar{c}$. By equation  \eqref{eq:cs}, $\lambda = 0$. Equation
\eqref{eq:foc} then gives $\mu = u'(c) = e^{-\gamma(c - \bar{c})}$. Differentiating it with respect to time, equating it with \eqref{eq:costate} and dividing by $-\gamma \mu \neq 0$, one obtains the second equation in \eqref{eq:ramsey_explicit}. The CARA specification has the distinctive property that the right hand side of the Euler equation does not involve
$c$ so that the derivative of consumption with respect to time only depends on the state $k$. This decoupling is the
structural feature that will enable the closed-form expression later. 

Assume now that consumption is constrained so that $c=\cbar$. Then from equation \eqref{eq:foc}, it follows that $\lambda = \mu -
u'(\bar{c}) = \mu - 1$ given that $e^{-\gamma\cdot 0}=1$. The complementary-slackness condition \eqref{eq:cs} requires
$\lambda \geq 0$ and thus $\mu\geq 1$. Substituting $c = \bar{c}$ into the resource constraint \eqref{eq:resource_constraint} gives $
\dotk = k^\alpha - (\delta + n) k - \bar{c}$, which is strictly positive for $k\in(k_{\min},k^*(\rho))$ \textemdash see Lemma \ref{lem:kmin}\emph{(iii)}-\emph{(iv)}. If initial capital $k_0$ is close to $k_{\min}$, capital therefore accumulates monotonically without any
consumption above subsistence until the economy reaches a capital level $\overline{k}$ at which the
interior Euler equation \eqref{eq:ramsey_explicit} becomes consistent with $c = \bar{c}$.\footnote{such a level exists for the auxiliary problem where $\rho=n$ by Lemma \ref{lem:threshold}, and for 
$\rho$ near $n$ by the implicit function theorem argument in the proof of Theorem \ref{thm:smooth_rho}.}

Observe now that the costate equation \eqref{eq:costate}, $\dot{\mu} = \bigl[\rho - r(k)\bigr]
\mu$, holds in both regimes: the Kuhn-Tucker multiplier $\lambda$ does not appear in
$\partial \mathcal{H}/\partial k = \mu[\alpha k^{\alpha-1} - (\delta + n)]$, because the subsistence
constraint is on the control $c$, not on the state $k$. Since $k$ is continuous in $t$ and
$r(k)$ is smooth, the right-hand side $\bigl[\rho - r(k) \bigr] \mu$ is continuous in
$t$, and therefore $\mu$ is continuously differentiable \textemdash in particular, continuous at
the junction time when $k$ crosses $\bar{k}(\rho)$.

At the junction point $\bar{k}(\rho)$ where the economy transitions from the constrained to
the interior regime, the interior first-order condition \eqref{eq:foc} with $\lambda =
0$ gives $\mu(\bar{k}(\rho)) = u'(\bar{c}) = 1$. By continuity of $\mu$, the same value
obtains from the constrained side: $\lambda(\bar{k}(\rho)) = \mu(\bar{k}(\rho)) - 1 = 0$. In the
constrained region where $k < \bar{k}(\rho) < k^{\ast}$,  $r$ is strictly decreasing in $k$
and satisfies $r(k) > r(k^{\ast}) = \rho$, so the costate equation gives $\dot{\mu} = \bigl[\rho - r(k) \bigr] \mu < 0$: the costate is strictly decreasing in forward
time for every $\rho > n$. Since $\mu = 1$ at the exit time (when $k = \bar{k}(\rho)$)
and $\mu$ is decreasing forward, $\mu(t) > 1$ for all earlier times $t$ in the
constrained regime. Therefore $\lambda(t) = \mu(t) - 1 > 0$ throughout, confirming
that the Kuhn-Tucker condition \eqref{eq:cs} is satisfied with strict inequality. In order to ensure that the transversality condition \eqref{eq:tvc} holds for $\rho>n$, I pick the solution which implies that savings converge to $s(k^\ast)=0$ at steady state. In practice, this means that the differential equation for the consumption/saving function will feature a terminal condition. 

\subsection{The optimal savings/consumption function}

Whenever $k_{min}<k<\overline{k}$, the optimal consumption function is $c(k)=\cbar$. Accordingly, I focus on the arguably more interesting case where the economy has left the subsistence regime so that $k>\kbar(\rho)$. Along the convergent saddle path\textemdash the trajectory selected by the terminal condition $s(k^\ast(\rho))=0$, which approaches the steady state $(k^\ast(\rho),0)$ from below\textemdash the economy can be represented for $k\in(\kbar(\rho),k^\ast(\rho))$ by the following vector field:
\begin{align}
\dot{x} = \begin{pmatrix}
\dotk\\
\dot{s}
\end{pmatrix}
=
\begin{pmatrix}
s \\
\left[r(k)-n\right]s - \frac{r(k)-\rho}{\gamma}
\end{pmatrix}
\defeq F(x;\rho)
\label{eq:vector_field}
\end{align}
where the second equation has been obtained by taking the time derivative of $s=\dotk$. In order to solve for the optimal saving function $s(k)$, the following result will be useful:
\begin{lemma}
\label{lem:kdot_pos}
Let $k\in (\kbar(\rho), k^{\ast}(\rho))$. Then it follows that $\dotk>0$.
\end{lemma}
\begin{proof}
Consider the optimal trajectory $(k(t),s(t))$ solving \eqref{eq:vector_field}, which converges to the steady state $(k^\ast(\rho),0)$. Since $\dotk=s$, it is enough to show that $s(t)>0$ for every $t$ on this arc. On the segment $\{s=0,\ k<k^\ast(\rho)\}$ the vector field \eqref{eq:vector_field} satisfies
\[
\dotk\big|_{s=0}=0,
\qquad
\dot{s}\big|_{s=0}=-\frac{r(k)-\rho}{\gamma}<0,
\]
where the inequality holds because $r$ is strictly decreasing with $r(k)>r(k^\ast(\rho))=\rho$ for $k<k^\ast(\rho)$. Thus at every point of this segment the flow points strictly downward, into $\{s<0\}$.
 
Suppose, for contradiction, that $s(t_0)=0$ at some finite time $t_0$ on the arc, with $k(t_0)<k^\ast(\rho)$. Since the orbit reaches $s=0$ only as $t\to\infty$, such a $t_0$ would be an interior zero, so $s$ would change sign there; but $\dot{s}(t_0)<0$ forces $s$ to be strictly decreasing through the zero, so $s>0$ immediately before $t_0$ and $s<0$ immediately after. Once $s<0$, we have $\dotk=s<0$, so $k$ is strictly decreasing thereafter and moves away from $k^\ast(\rho)$; the orbit can then no longer return to $(k^\ast(\rho),0)$, contradicting convergence. Hence $s(t)>0$ for all $t$ on the arc, and $\dotk=s>0$ on $(\kbar(\rho),k^\ast(\rho))$.
\end{proof}
This result implies that one can write both optimal saving and consumption as a function of $k$ as follows:
\begin{align}
s(k) \defeq f(k) - (\delta+n)k - c(k)
\label{eq:saving-def}
\end{align}
Indeed, given that $\dotk=s(k)>0$ on $(\kbar(\rho),k^\ast(\rho))$ by Lemma \ref{lem:kdot_pos}, the map $t\mapsto k(t)$ is a strictly increasing $C^1$ bijection so consumption may be written as a function of $k$ alone with $c'(k)=\dotc/\dotk$\textemdash the Time-Elimination Method of \cite{Mulligan1991note}. Equation \eqref{eq:saving-def} is only implicit at the moment given that it depends on the unknown function $c(k)$. In this context, one can show that the equilibrium saving function is the solution of a particular ODE. This is established in the following lemma:
\begin{lemma}[Saddle path ODE] \label{lem:ode}
On any open subinterval of $(\kbar(\rho), k^{\ast}(\rho))$, the saving function $s(k)$ satisfies the following first-order ordinary differential equation:
\begin{equation} \label{eq:saddle-ode}
s'(k)s(k)= (r(k) - n) s(k) - \frac{r(k) - \rho}{\gamma},
\end{equation}
together with the terminal condition $\lim_{k \to k^{\ast}} s(k) = 0$.
\end{lemma}
 
\begin{proof}
Using Lemma \ref{lem:kdot_pos} as well as the chain rule gives $\dot{s} = s'(k)\dotk=s'(k)s(k)$. Using this and differentiating \eqref{eq:saving-def} along the optimal trajectory using the chain rule again, I obtain:
\begin{align}
\nonumber
s'(k)s(k)  &= \bigl[f'(k) - (\delta + n)\bigr] \dotk - c'(k)\dotk\\
&= \bigl[f'(k) - (\delta + n)\bigr] \dotk - \dotc\notag\\
&= (r(k) - n) s(k) - \frac{r(k) - \rho}{\gamma},
\label{eq:ODE_sk}
\end{align}
where I have used the definition of the real rate of return \eqref{eq:rk_def} together with the Euler equation \eqref{eq:ramsey_explicit} to substitute for $\dotc$. The terminal condition follows from the fact that
$\lim_{t \to \infty} s(k(t)) = 0$ on the convergent saddle path to ensure that the transversality condition \eqref{eq:tvc} holds.
\end{proof}

Equation \eqref{eq:ODE_sk} is an Abel ODE of the second kind. Unfortunately, there is no known closed form solution for such an equation.\footnote{Some Abel ODEs of the second kind are amenable to a closed-form solution \textemdash see \cite{Mhadhbi2024exact}. Unfortunately, these require a structure that doesn't match the specific example considered in this paper.} In the special case where $\rho=n$, equation \eqref{eq:ODE_sk} becomes separable and thus amenable to a closed form solution. The problem is that under this parameter restriction, the maximization objective \eqref{eq:disc_util} diverges to $+\infty$ and the transversality condition \eqref{eq:tvc} doesn't hold anymore. With this in mind, I will take the following approach: solve the ODE \eqref{eq:ODE_sk} explicitly for the case $\rho=n$ and then show that the solution of the ODE is smooth in $\rho$ which allows one to consider the case $\rho\to n^+$. In this context, I define $k^*(n)$ as the first part of equation \eqref{eq:steady_states} evaluated at $\rho=n$. 

\subsection{An explicit expression for the savings function}

Working directly with $s\defeq\dotk$ and using $\ddot k = s\,\mathrm{d}s/\mathrm{d}k$, taking the time derivative of the accumulation equation and substituting the Euler equation \eqref{eq:ramsey_explicit} with $\rho=n$ gives:
\begin{align}
s\,\frac{\mathrm{d}s}{\mathrm{d}k} = \bigl(\alpha k^{\alpha-1}-\delta-n\bigr)\Bigl(s-\tfrac{1}{\gamma}\Bigr)
\iff \Bigl(1+\frac{1/\gamma}{\,s-1/\gamma\,}\Bigr)\mathrm{d}s = \bigl(\alpha k^{\alpha-1}-\delta-n\bigr)\mathrm{d}k.
\label{eq:ODE_rho_n}
\end{align}
It can be shown (see Lemma \ref{lem:no-obstruction}\emph{(iii)} in Appendix \ref{sec:app_proof_smooth_rho}) that $s(k)\in(0,1/\gamma)$ on $(k_{\min},k^\ast(n))$. Using this to integrate \eqref{eq:ODE_rho_n}, one gets:
\begin{align*}
s(k) + \frac{1}{\gamma}\ln\left(\tfrac{1}{\gamma}-s(k)\right) = \psi(k) + A,
\end{align*}
where, as before, $\psi(k)= f(k)-(\delta+n)k$ and $A$ is a constant of integration. Imposing $s(k^\ast(n))=0$ pins down $A=\ln(1/\gamma)/\gamma-\psi(k^\ast(n))$, which now gives:
\begin{align}
\gamma s(k)+\ln(1-\gamma s(k)) = \gamma\bigl(\psi(k) - \psi(k^\ast(n))\bigr).    
\label{eq:gamma_sk_Lambert}
\end{align}
Defining $u(k)\defeq \gamma s(k)-1$ gives an expression which, after exponentiating both sides, can be inverted with the first branch\footnote{The relevant branch is $W_0$ and not $W_{-1}$ because $s(k)\in(0,1/\gamma)$ forces $\gamma s(k)-1 \in (-1,0)$ which is the range of $W_0$; see \cite{Mezo2022lambert} for a textbook treatment and \cite{Roulleau2026explicit} for the same construction.} of the Lambert $W$ function denoted by $W_0$. This yields the following proposition:

\begin{prop}
\label{prop:closed_form}
Assume that $k\in(\bar{k}(n),k^\ast(n))$. It follows that the solution of the ODE \eqref{eq:ODE_rho_n} is:
\begin{align}
s(k;n) = \frac{1+w_0(k)}{\gamma},\quad w_0(k)\defeq W_0\left(\nu\cdot \exp\left\{\gamma(\psi(k))\right\}\right)   ,
\end{align}
where I have defined $\nu\defeq -\exp\left\{-1-\gamma(\psi(k^\ast(n)))\right\}$. The associated consumption function is $c(k;n) = f(k) - (\delta+n)k-s(k;n)$.
\end{prop}

With these explicit expressions in hand, one can now study the properties of the savings/consumption function. These are described in the following corollary:
\begin{corollary} 
\label{cor:properties}
Assume $\rho = n$. The consumption function on the interior interval $(\kbar(n), k^{\ast}(n))$ satisfies the following properties:
\begin{align*}
(i)\quad  s'(k)<0, \quad (ii)\quad c'(k) > 0, \quad (iii)\quad c''(k) < 0  
\end{align*}
\end{corollary}
\begin{proof}
See Appendix \ref{sec:app_proof_properties}. 
\end{proof}

It then follows from Corollary \ref{cor:properties} that the saving function is strictly decreasing, and facts (ii) and (iii) together establish that the consumption function is strictly increasing and \textit{concave}. This is in contrast with the CRRA case in which the stable arm for consumption can either be concave or convex depending on the degree of risk aversion. That being said, \cite{Barro2004economic} argue that the concave one is the empirically relevant case.

\subsection{A $O(\rho-n)-$accurate approximation}

With this explicit closed form expression in hand, the objective is to show that as $\rho\to n^+$ the true solution will be arbitrarily close to this explicit expression. This is guaranteed by the following theorem:
\begin{theorem}
\label{thm:smooth_rho}
Let $[k_{1}, k_{2}]$ be a closed interval with $\kmin \le k_{1} < k_{2} < k^{\ast}(n)$ and
$\bar{k}(n) \notin [k_{1}, k_{2}]$, and let the limiting policy be defined piecewise by
\begin{align}
s(k;n)\defeq\begin{cases}
\psi(k)-\cbar, & k\in[\kmin,\bar{k}(n)),\\[2pt]
\dfrac{1+w_0(k)}{\gamma}, & k\in(\bar{k}(n),k^{\ast}(n)),
\end{cases}
\label{eq:piecewise_limit}
\end{align}
that is, by the constrained policy $c=\cbar$ below the take-off point and by the closed form
of Proposition \ref{prop:closed_form} above it. Then there exist $\epsilon > 0$ and a constant $C
\geq 0$ (both depending on $[k_{1}, k_{2}]$ and on the parameters $\gamma, \alpha, \delta, n$)
such that, for every $\rho$ satisfying $0 < \rho - n < \epsilon$, the interval $[k_{1},k_{2}]$
lies entirely within one regime of the problem indexed by $\rho$---either
$[k_{1},k_{2}]\subset[\kmin,\bar{k}(\rho))$ or $[k_{1},k_{2}]\subset(\bar{k}(\rho),k^{\ast}(\rho))$,
according to the side of $\bar{k}(n)$ on which it lies---so that $s(k;\rho)$ is defined on it, and
\begin{equation} \label{eq:bound}
\sup_{k \in [k_{1}, k_{2}]} \bigl| s(k; \rho) - s(k;n) \bigr| \leq
C (\rho - n),
\end{equation}
where $s(k;\rho)$ is the saving function that solves the maximization program
\eqref{eq:disc_util}-\eqref{eq:subsistence_constraint}. An analogous bound holds for the
associated consumption function $c(k;\rho)$.
\end{theorem}
\begin{proof}
See Appendix \ref{sec:app_proof_smooth_rho}.
\end{proof}
Note that I have chosen $\bar{k}(n)=\lim_{\rho\to n^+}\bar{k}(\rho)$ and $k^\ast(n)=\lim_{\rho\to n^+}k^\ast(\rho)$ as endpoints so that they do not vary with $\rho$, which allows me to state the uniform convergence property on a fixed compact set. Since $\bar{k}(n)\notin[k_1,k_2]$ and $k_2<k^\ast(n)$ with strict inequalities, and since $\bar{k}(\rho)\to\bar{k}(n)$ and $k^\ast(\rho)\to k^\ast(n)$ as $\rho\to n^+$, choosing $\epsilon$ small enough ensures that $[k_1,k_2]$ remains on the same side of $\bar{k}(\rho)$ as of $\bar{k}(n)$, and that $k_2<k^\ast(\rho)$. The saving and consumption functions are therefore well defined on $[k_1,k_2]$, and the two regimes of \eqref{eq:piecewise_limit} are precisely the two cases distinguished in the proof.

With this in mind, Theorem \ref{thm:smooth_rho} states that the closed-form solution for the auxiliary problem $\rho=n$ provides a uniformly $O(\rho-n)$-accurate approximation to the true saddle path on compact subsets of the interior state space; the error can be made arbitrarily small by taking $\rho\to n^+$. This justifies focusing on the case $\rho=n$ in the remainder of the paper. The consumption function from Proposition \ref{prop:closed_form} is illustrated in the following figure, plotted both as a function of capital and time. 

\begin{figure}[ht]
	\centering
	\bigskip
	\caption{From subsistence to steady state consumption}

	{\small
\includegraphics[width=\textwidth]{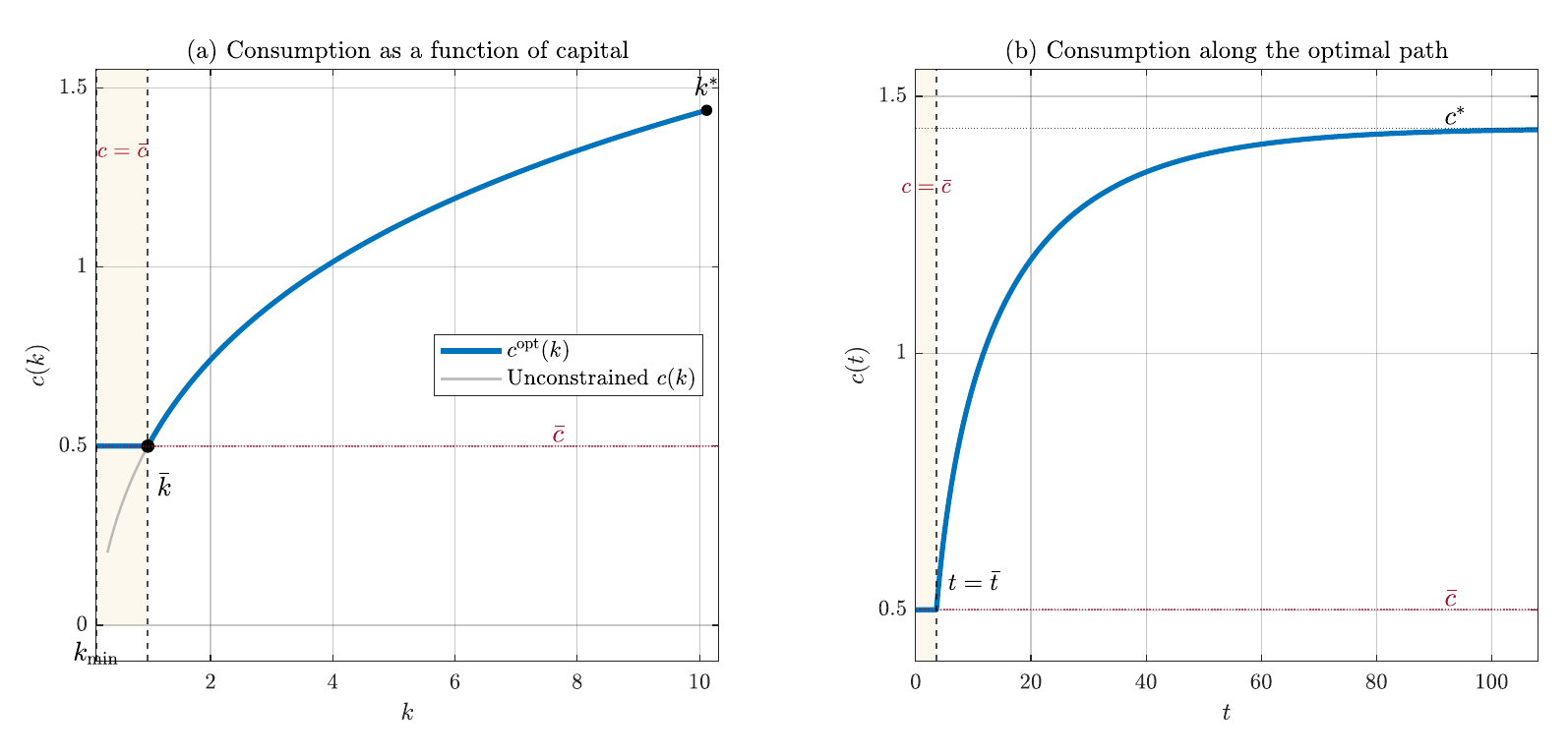}
	}
	\begin{minipage}{0.75\linewidth}
	{
 	 \footnotesize\emph{Notes: parameter values $\alpha=1/3$, $\gamma=2$, $\delta=0.05$, $\rho=n=0.02$, and $\cbar=0.5$.} 
	 }
	\end{minipage}
    \label{fig:Consumption_Function}
\end{figure}

Given the initial value of capital, that economy doesn't have enough resources to consume above and beyond the subsistence level. Therefore, for a small time period it consumes at the subsistence level. Once it has accumulated enough capital, this economy experiences a take-off and starts growing until it converges to its steady state. At this point, one might wonder how fast the economy is growing along the trajectory. In particular, one might wonder whether the convergence speed is faster close to $\cbar$ or to $c^\ast$. While common intuition would suggest that the speed of convergence is faster during the initial take-off, this is not necessarily the case here. Indeed, it turns out that the degree of risk aversion is crucial in shaping the speed of convergence. This is detailed in the following proposition.

\begin{prop}
\label{prop:threshold_gamma}
Let $\rho=n$ and set $\Delta^{\ast}\defeq\psi(k^{\ast}(n))-\cbar=c^{\ast}-\cbar>0$. At the
take-off point, the saving function admits the elementary closed form
\begin{align}
s(\bar{k}(n);n)=\frac{1-e^{-\gamma\Delta^{\ast}}}{\gamma},
\qquad\text{so that}\qquad
\psi(\bar{k}(n))=\cbar+\frac{1-e^{-\gamma\Delta^{\ast}}}{\gamma}.
\label{eq:junction_closed_form}
\end{align}
Let $\beta(k;\gamma)\defeq s(k;n)/(k^{\ast}(n)-k)$ denote the speed of convergence to
$k^{\ast}(n)$, with $\beta(k^{\ast};\gamma)$ its limit as $k\uparrow k^{\ast}(n)$, and define
$R(\gamma)\defeq\beta(\bar{k}(n);\gamma)/\beta(k^{\ast};\gamma)$. Then
$\beta(k^{\ast};\gamma)=\sqrt{|f''(k^{\ast}(n))|/\gamma}$ and
\begin{enumerate}
\item[(i)] $\displaystyle\lim_{\gamma\to0^{+}}R(\gamma)=1^{+}$
\item[(ii)] $\displaystyle\lim_{\gamma\to\infty}R(\gamma)=0$.
\end{enumerate}
Consequently, since $R$ is continuous, the Intermediate Value Theorem guarantees the existence of at least one $\underline\gamma>0$ with $R(\underline\gamma)=1$. Convergence is therefore front-loaded for all sufficiently small values of $\gamma$ and back-loaded for all sufficiently large values.
\end{prop}
\begin{proof}
See Appendix \ref{sec:app_proof_threshold_gamma}.   
\end{proof}
In this proposition, $\beta(k^\ast)$ is computed using the first order approximation of the saving function around the steady state $k^\ast$. As a result, it is a direct counterpart to the measure derived in closed form for a CRRA utility function in \cite{Barro2004economic}. The explicit closed-form expression developed in the current paper allows one to compare it to
the speed that applies at the take-off point $\bar{k}(n)$, that is, at the earliest stage of
development at which the economy leaves the subsistence regime. 

In the first case, the consumption path is front-loaded. If risk-aversion is low enough, then this economy devotes a significant fraction of its resources to capital accumulation, which makes it initially fast. In the second case, risk aversion is relatively high. As a result, this economy accumulates capital relatively slowly. When the stock of capital has reached a certain point close to the steady state, the accumulation becomes faster.

\section{Conclusion}

In this paper, I have developed a closed form expression for the consumption/savings function in the special case of a neoclassical growth model where the utility function has a CARA specification and the discount factor is equal to the growth rate of population. I have shown that it provides a uniformly $O(\rho-n)$-accurate approximation of the solution on compact subsets of the interior state space as $\rho\to n^+$. 

Using this closed form expression, I have derived properties of the consumption/saving function as well as for the speed of convergence both close to and far from the steady state. This approach could be used in teaching in order to exemplify a case where one can construct the stable arm of a neoclassical growth model in closed form. Alternatively, one can use the explicit solution derived in this paper as a benchmark to assess the accuracy of recent Deep Learning-based solution methods as in \cite{Fernandez2026deep} and \cite{Judd2026deep}, which both consider numerical solutions of a neoclassical growth model with CRRA preferences. 


\appendix

\section{Proof of Lemma \ref{lem:kmin}}
\label{sec:app_proof_Lemma_kmin}

By definition of $\psi(k)$, it follows that $\psi'(k)=\alpha k ^{\alpha-1}-(\delta+n)$ and $\psi''(k)=\alpha(\alpha-1)k^{\alpha-2}$. Given that $\alpha<1$, $\psi''(k)<0$ and $\psi$ is strictly concave on $(0,\infty)$. It is straightforward to observe that $\psi(0)=0$. In addition, from the expression of $\psi'(k)$ it follows that $\psi'(k)>0$ iff $k<k_{\psi}\defeq (\alpha/(\delta+n))^{1/(1-\alpha)}$ and vice versa for $\psi'(k)<0$ so that $\psi$ is strictly increasing on $(0,k_{\psi})$ and strictly decreasing on $(k_{\psi},\infty)$. As a result, $\psi$ attains its maximum at $k_{\psi}$. Furthermore, $k^\alpha/((\delta+n)k)\to 0$ as $k\to\infty$ so that the linear term dominates for large $k$ and $\psi(k)\to -\infty$ as $k\to\infty$. This proves part \emph{(i)}.

Given that $(\alpha/(\delta+\rho))^{1/(1-\alpha)}$ is clearly decreasing in $\rho$ and that $\rho>n$ by assumption, it follows that $k^*<k_{\psi}$. This proves part \emph{(ii)}.

The properties of $\psi(k)$ imply that $\psi(k)-\cbar$ is also strictly increasing on $(0,k_{\psi})$. Furthermore, $\psi(k^*)=c^*>\cbar$ by assumption so that $\psi(k^*)-\cbar>0$. Given that $\psi(k)-\cbar$ attains its maximum at $k_{\psi}$, it follows that $\psi(k_{\psi})-\cbar>0$. Therefore, there exists a value $k_{min}\in (0,k_{\psi})$ such that $\psi(k_{min})=\cbar$. For $k\in (k_{\psi},\infty)$, $\psi(k)-\cbar$ is strictly decreasing and goes to $-\infty$. Given that it starts out strictly positive at $k_{\psi}$, it follows that there exists a value $k_{max}\in (k_{\psi},\infty)$ such that $\psi(k_{max})=\cbar$. It follows immediately that $\psi(k)-\cbar>0$ on $(\kmin,\kmax)$ and $\psi(k)-\cbar<0$ outside $[\kmin,\kmax]$. This proves part \emph{(iii)}.

Given the results in \emph{(iii)} and the fact that $\psi(k^*)-\cbar>0$, it follows that $k_{min}<k^*(\rho)<k_{max}$. This proves part \emph{(iv)}.

\section{Proof of Corollary \ref{cor:properties}}
\label{sec:app_proof_properties}
 
For brevity, in this proof I drop the $n$-dependence and
write $k^{\ast}$, $s(k)$ for $k^{\ast}(n)$ and
$s(k;n)$ respectively, and I set $q(k) \defeq \gamma\, s(k)$, so that the closed form of Proposition \ref{prop:closed_form} reads $q(k)+\ln(1-q(k))=\gamma\bigl(\psi(k)-\psi(k^\ast)\bigr)$ with $\psi(k):=f(k)-(\delta+n)k$. I use the saddle-path ODE \eqref{eq:ODE_rho_n}
in the form
\begin{equation*}
s(k)s'(k) = \psi'(k) \left[ s(k) - \tfrac{1}{\gamma} \right]
\end{equation*}
together with Lemma \ref{lem:no-obstruction} \emph{(iii)}, which establishes that $s(k)
\in (0, 1/\gamma)$, or equivalently that $q(k) \in (0, 1)$, on $(\kbar(n), k^{\ast})$.

\emph{(i) Monotonicity of $s(k)$.} Notice that:
\begin{align*}
s'(k) = \frac{\psi^{'}(k)}{s(k)}\left[s(k)-\frac{1}{\gamma}\right]  \end{align*}
Given that $s(k)\in(0,1/\gamma)$ on $k\in(\bar{k}(n),k^\ast)$, it immediately follows that $s'(k)<0$ so that savings are a strictly decreasing function of capital. 
  
\medskip
\emph{(ii) Monotonicity of $c$.} By the resource constraint, $c(k) = f(k) - (\delta+n)k -
s(k)$, hence
\begin{equation*}
c'(k) = \psi'(k) - s'(k) = \psi'(k) -\frac{\psi'(k)}{s(k)}\left[s(k)-\frac{1}{\gamma}\right]
= \frac{\psi'(k)}{q(k)} = \frac{r(k) - n}{\gamma s(k)}.
\end{equation*}
On $(\kbar(n), k^{\ast})$, $\psi'(k) > 0$ and $q(k) > 0$, hence $c'(k) > 0$.
 
\emph{(iii) Concavity of $c$.} The closed-form solution is given by the implicit
relation $q(k) + \ln(1 - q(k)) = \gamma(\psi(k) - \psi(k^{\ast}))$. Differentiating
both sides with respect to $k$ gives
\begin{equation*}
\left[1 - \frac{1}{1 - q(k)}\right] q'(k) = -\frac{q(k)}{1 - q(k)} q'(k) = \gamma\psi'(k),
\end{equation*}
which yields
\begin{equation} \label{eq:qprime}
q'(k) = - \frac{\gamma \psi'(k) (1 - q(k))}{q(k)}.
\end{equation}
Since $c(k) = \psi(k) - q(k)/\gamma$, we have $c'(k) = \psi'(k) - q'(k)/\gamma$
and $c''(k) = \psi''(k) - q''(k)/\gamma$. Differentiating \eqref{eq:qprime}
yields:
\begin{align*}
q''(k) &= -\gamma\left\{\psi^{''}(k)\frac{1-q(k)}{q(k)}+\psi^{'}(k)\frac{-q'(k)q(k)-q'(k)(1-q(k))}{(q(k))^2}\right\}\\
&= -\gamma\left\{\psi^{''}(k)\frac{1-q(k)}{q(k)}-\psi^{'}(k)\frac{q'(k)}{(q(k))^2}\right\}\\
&= -\gamma\left\{\psi^{''}(k)\frac{1-q(k)}{q(k)}+\gamma(\psi^{'}(k))^2\frac{1-q(k)}{(q(k))^3}\right\}\\
&= -\gamma\frac{1-q(k)}{q(k)}\left\{\psi^{''}(k)+\gamma\frac{(\psi^{'}(k))^2}{(q(k))^2}\right\},
\end{align*}
where I have used equation \eqref{eq:qprime} to substitute for $q'(k)$ on the third line. It follows that
\begin{equation*}
c''(k) = \psi''(k) - \frac{q''(k)}{\gamma} = \frac{\psi''(k)}{q(k)} + \frac{\gamma \psi'(k)^{2} (1 - q(k))}{q(k)^{3}}.
\end{equation*}
Multiplying by $q(k)^{3}$, which is strictly positive on $(\kbar(n), k^{\ast})$,
\begin{equation} \label{eq:V-def}
V(k) \defeq q(k)^{3} c''(k) = \psi''(k) q(k)^{2} + \gamma \psi'(k)^{2} (1 - q(k)).
\end{equation}
The sign of $c''(k)$ coincides with the sign of $V(k)$. The proof proceeds by
showing that $V(k^{\ast}) = 0$ and that $V$ is strictly increasing on $(\kbar(n),
k^{\ast})$, so that $V(k) < V(k^{\ast}) = 0$ for $k < k^{\ast}$.
 
\smallskip
\emph{Step 1: $V(k^{\ast}) = 0$.} At $k = k^{\ast}$, both $q(k^{\ast}) = 0$ and
$\psi'(k^{\ast}) = r(k^{\ast}) - n = 0$. Substituting into \eqref{eq:V-def} gives
$V(k^{\ast}) = 0$.
 
\smallskip
\emph{Step 2: $V'(k) > 0$ on $(\kbar(n), k^{\ast})$.} Differentiating
\eqref{eq:V-def},
\begin{equation*}
V'(k) = \psi'''(k) (q(k))^{2} + 2\psi''(k) q(k) q'(k) + 2\gamma \psi'(k) \psi''(k) (1 - q(k)) - \gamma (\psi'(k))^{2} q'(k).
\end{equation*}
Substituting $q'(k) = -\gamma \psi'(k)(1 - q(k))/q(k)$ from \eqref{eq:qprime} into the
second and third terms so that they cancel out, one is left with:
\begin{equation} \label{eq:Vprime}
V'(k) = \psi'''(k) q(k)^{2} + \gamma^{2}\frac{1 - q(k)}{q(k)}\psi'(k)^{3} .
\end{equation}
One can now check that both terms in \eqref{eq:Vprime} are strictly positive on
$(\kbar(n), k^{\ast})$:
\begin{itemize}
\item \emph{Third derivative of $\psi$.} For a Cobb-Douglas production function $f(k) = k^{\alpha}$,
\begin{equation*}
\psi'''(k) = f'''(k) = \alpha(\alpha - 1)(\alpha - 2) k^{\alpha - 3}.
\end{equation*}
Since $\alpha \in (0, 1)$, both factors $(\alpha - 1)$ and $(\alpha - 2)$ are
negative so that their product is positive. Combined with $\alpha > 0$ and
$k^{\alpha - 3} > 0$, this gives $\psi'''(k) > 0$. With $q(k)^{2} > 0$, the
first term in \eqref{eq:Vprime} is strictly positive.
 
\item \emph{First derivative of $\psi$.} On $(\kbar(n), k^{\ast})$, $\psi'(k) =
r(k) - n > 0$, hence $\psi'(k)^{3} > 0$. With $\gamma > 0$, $1 - q(k) > 0$,
and $q(k) > 0$, the second term in \eqref{eq:Vprime} is strictly positive.
\end{itemize}
Therefore $V'(k) > 0$ on $(\kbar(n), k^{\ast})$. Combining Steps 1 and 2, $V(k) < V(k^{\ast}) = 0$ for
all $k \in (\kbar(n), k^{\ast})$. Since $q(k)^{3} > 0$, one can conclude that $c''(k) =
V(k)/q(k)^{3} < 0$ on $(\kbar(n), k^{\ast})$.
 
This completes the proof. \hfill $\blacksquare$

\section{Proof of Theorem \ref{thm:smooth_rho}}
\label{sec:app_proof_smooth_rho}

I now establish that the closed-form policy $(s(k;n), c(k;n)$ is the
uniform limit of the constrained optimal policies $(s(k;\rho), c(k;\rho))$ as
$\rho \to n^+$, on compact subsets of the interior state space excluding a
neighborhood of the kink.

\subsection{Setup and Smooth Dependence}

Throughout this appendix the vector field is smooth on the interior (Lemma \ref{lem:no-obstruction}(ii)), so all manifold and flow results are stated for smooth vector fields. I first record, in the paper's own notation, the four results from the theory of dynamical systems that are used below, so that the reader need not consult \cite{Wiggins2003introduction} in parallel. \emph{(P1) Local invariant manifolds:} a hyperbolic fixed point of a smooth vector field has smooth local stable and unstable manifolds, tangent to the corresponding eigenspaces and representable there as graphs (Theorem 3.2.1 in \cite{Wiggins2003introduction}). \emph{(P2) Global continuation:} the global stable manifold is the transport of the local one along the flow, $W^s=\bigcup_{t\ge0}x(-t,x_0;\rho)$ with $x_0\in \mathrm{W^s_{loc}}$ and where $x(t,x_0;\rho)$ denotes the solution of the dynamical system at time $t$ (going backwards) for initial condition $x_0$ in the local stable manifold $W^s_{loc}$. It is invariant and inheriting smoothness because trajectories of a smooth vector field are smooth in their initial conditions (\S3.5, eq.3.5.13 in \cite{Wiggins2003introduction}; stated there for the unstable manifold, the stable case following by time reversal). \emph{(P3) Continuation up to the boundary:} a solution extends uniquely until it reaches the boundary of any compact set containing it (Theorem 7.2.1 in \cite{Wiggins2003introduction}). \emph{(P4) Smooth dependence:} the flow $x(t,x_0;\rho)$ is jointly smooth in $t$, $x_0$ and $\rho$ (Theorem 7.3.1 in \cite{Wiggins2003introduction}). The relevant vector field is reproduced here for convenience:
\begin{align}
\dot{x} = \begin{pmatrix}
\dotk\\
\dot{s}
\end{pmatrix}    
=
\begin{pmatrix}
s \\
\left[r(k)-n\right]s - \frac{r(k)-\rho}{\gamma} & 
\end{pmatrix}
\defeq F(x;\rho)
\label{eq:vector_field_app}
\end{align}
For each $\rho > n$, the planner's problem admits a unique optimal saddle path. Denote
the corresponding optimal saving and consumption functions by $s(k; \rho)$ and $c(k;
\rho)$. The threshold at which the constraint $c \geq \bar{c}$ begins to bind is denoted
$\bar{k}(\rho)$, and the steady state is $k^{*}(\rho)$. Using standard results in the theory of hyperbolic dynamical systems with parameters, one can show that the global solution for $\rho>n$ varies smoothly with $\rho$. Before stating the result and its proof formally, it needs to be verified first that the region traversed by the saddle path contains no degeneracies that would obstruct such a global approximation. This is established in the following Lemma.

\begin{lemma}[Absence of obstructions along the saddle path] \label{lem:no-obstruction}
For every $\rho$ in a neighbourhood of $n$, the system \eqref{eq:vector_field} has the
following properties on the half-plane $\{(k, s) : k > 0\}$:
\begin{enumerate}
\item[(i)] The only equilibrium is $(k^{*}(\rho), 0)$, which is a hyperbolic saddle;
\item[(ii)] The vector field is $C^{\infty}$ on every compact subset of the interior;
\item[(iii)] The saddle path is bounded: $0 < s(k; \rho) \leq 1/\gamma$ for every $k \in
(\bar{k}(\rho), k^{*}(\rho))$;
\item[(iv)] The saddle path admits no homoclinic connection: stable and unstable
eigendirections at $(k^{*}(\rho), 0)$ are distinct one-dimensional manifolds that do
not reconnect.
\end{enumerate}
\end{lemma}

\begin{proof}
\emph{(i)} Equilibria require $s = 0$ and $r(k) = \rho$. The Inada conditions on $f$
ensure that $r : (0, \infty) \to \mathbb{R}$ is strictly decreasing with range $(-\delta,
\infty)$, so $r(k) = \rho$ has a unique positive solution $k^{*}(\rho)$. The Jacobian of
the right-hand side of \eqref{eq:vector_field} at $(k^{*}(\rho), 0)$ is
\begin{equation*}
\mathbf{J}=
\begin{pmatrix} 0 & 1 \\ -\frac{f''(k^{*}(\rho))}{\gamma} & \rho - n \end{pmatrix},
\end{equation*}
where I have used $r(k^\ast(\rho))=\rho$ at the steady state. This implies that $\det{\mathbf{J}} = f''(k^{*}(\rho))/\gamma < 0$ (since $f'' < 0$ by concavity), giving
real eigenvalues of opposite sign. Hence $(k^{*}(\rho), 0)$ is a hyperbolic saddle for $\rho>n$ \textemdash see Definition 1.2.6 in \cite{Wiggins2003introduction}.

\emph{(ii)} Immediate from $f(k) = k^{\alpha}$ being $C^{\infty}$ on $(0, \infty)$.

\emph{(iii)} \noindent The lower bound is immediate from Lemma \ref{lem:kdot_pos}. In order to derive the upper bound, write the ODE in normal form (valid for $s > 0$):
\begin{equation}\label{eq:normal}
  s'(k) = \bigl(r(k) - n\bigr) - \frac{r(k) - \rho}{\gamma s(k)}.
\end{equation}
I distinguish two cases.\\
\noindent\emph{Case $\rho > n$.}
Evaluating~\eqref{eq:normal} at any point $k_0$ where $s(k_0) = 1/\gamma$ gives
\[
  s'(k_0) = \bigl(r(k_0) - n\bigr) - \bigl(r(k_0) - \rho\bigr) = \rho - n > 0.
\]
Now suppose, for contradiction, that $s$ attains the value $1/\gamma$ somewhere on $(\bar{k}(\rho), k^*)$. Since $s(k^*) = 0 < 1/\gamma$, there exists a rightmost such point: let
\[
  k_0 := \sup\bigl\{ k \in (\bar{k}(\rho), k^*) : s(k) = 1/\gamma \bigr\}.
\]
Then $s(k_0) = 1/\gamma$ and $s(k) < 1/\gamma$ on $(k_0, k^*)$ because if not it would have to pass through $1/\gamma$ again from above, which is ruled out by the fact that $k_0$ is the supremum. Therefore, $s$ is non-increasing at $k_0$ from the right, giving $s'(k_0) \leq 0$. This contradicts $s'(k_0) = \rho - n > 0$.

\noindent\emph{Case $\rho = n$.}
The saving ODE becomes
\[
  ss' = \psi'(k)\Bigl(s - \frac{1}{\gamma}\Bigr),
\]
where $\psi'(k) := r(k) - n$. One can check directly that the constant function $\hat{s}(k) \equiv 1/\gamma$ satisfies this equation on $(\bar{k}(n), k^*)$: substituting gives $\tfrac{1}{\gamma}\cdot 0 = \psi'(k)\cdot 0$, which holds identically.\footnote{This comparison argument is available only in the auxiliary case: for $\rho>n$, $\hat s\equiv1/\gamma$ is not a solution, since substituting into the normal form \eqref{eq:normal} gives $s'=\rho-n\neq0$. That is why the case $\rho>n$ requires the separate crossing argument above.}
 
The right-hand side of the ODE written as $s' = \psi'(k)(s - 1/\gamma)/s$ is locally Lipschitz in $s$ for $s > 0$. By the uniqueness theorem 7.1.1 in \cite{Wiggins2003introduction}, no two distinct solutions can intersect. Since the saddle-path solution satisfies $s(k^*) = 0 \neq 1/\gamma = \hat{s}(k^*)$, it is a different solution from~$\hat{s}$ and therefore cannot touch $1/\gamma$ at any point. Combined with the lower bound and $s(k^*) = 0 < 1/\gamma$, this gives $s(k; n) \in (0, 1/\gamma)$ on $(\bar{k}(n), k^*)$.

\noindent
In both cases, $0 < s(k;\rho) < 1/\gamma$ on $(\bar{k}(\rho), k^*(\rho))$.

\emph{(iv)} Since $\det\mathbf{J}<0$, the eigenvalues satisfy $\lambda_s<0<\lambda_u$ with unstable eigenvector $\propto(1,\lambda_u)$, so the two branches of the local unstable manifold $W^u_{\mathrm{loc}}$ leave $(k^*(\rho),0)$ into $\{k>k^*,s>0\}$ and $\{k<k^*,s<0\}$.

Consider the branch $W^u_+$ and let $(k(t),s(t))$ be an orbit on it with $k(0)>k^*(\rho)$, $s(0)>0$. I claim $s(t)>0$ for all $t\ge0$. If not, let $t_0:=\inf\{t>0:s(t)=0\}$, finite by assumption. On $[0,t_0)$, $s>0$ so $\dotk=s>0$ and $k$ is increasing; by continuity $k(t_0)\ge k(0)>k^*(\rho)$, whence $r(k(t_0))<\rho$ and $\dot s(t_0)=-(r(k(t_0))-\rho)/\gamma>0$. But $s>0$ on $[0,t_0)$ with $s(t_0)=0$ forces $\dot s(t_0)\le0$, a contradiction. Hence $s(t)>0$ and $k(t)$ is strictly increasing for all $t\ge0$, so $k(t)\ge k(0)>k^*(\rho)$ and the orbit cannot converge to $(k^*(\rho),0)$.

The branch $W^u_-$ is symmetric: if $t_0:=\inf\{t>0:s(t)=0\}$ were finite, then $s<0$ and $\dotk<0$ on $[0,t_0)$, so $k(t_0)\le k(0)<k^*(\rho)$, whence $r(k(t_0))>\rho$ and $\dot s(t_0)=-(r(k(t_0))-\rho)/\gamma<0$; but $s<0$ before $t_0$ with $s(t_0)=0$ requires $\dot s(t_0)\ge0$. Hence $s(t)<0$ and $k(t)$ is strictly decreasing for all $t\ge0$, and the orbit leaves through $k=0$ without returning.

It remains to conclude that $W^s$ and $W^u$ meet only at $x^\ast(\rho)$. Suppose they shared a regular point $p\neq x^\ast(\rho)$. The orbit through $p$ then converges to $x^\ast(\rho)$ as $t\to+\infty$ (since $p\in W^s$) and also as $t\to-\infty$ (since $p\in W^u$); by uniqueness of solutions through the regular point $p$ (see theorem 7.1.1 in \cite{Wiggins2003introduction}) these are the same orbit, which is therefore homoclinic to $x^\ast(\rho)$. But such an orbit would have to leave $x^\ast(\rho)$ along a branch of $W^u$ and return to it, whereas the argument above shows that along each branch $k(t)$ is strictly monotone and the orbit escapes the interior\textemdash increasing to large $k$ along $W^u_+$, decreasing to the boundary $\{k=0\}$ along $W^u_-$\textemdash so no orbit returns to $x^\ast(\rho)$. Hence no homoclinic orbit exists. Finally, the manifolds cannot coincide along an arc through $x^\ast(\rho)$ itself, since near the equilibrium they are tangent to the distinct eigenspaces $E^s(\rho)\neq E^u(\rho)$ (distinct because $\lambda_s\neq\lambda_u$). Being one-dimensional and transverse at $x^\ast(\rho)$, and admitting no regular-point intersection, $W^s$ and $W^u$ meet only at $x^\ast(\rho)$.
\end{proof}

\begin{remark}
The absence of homoclinic orbits or other such irregularities means that one can start from a "seed" on the stable manifold around the steady state and construct the flow backward from there. The rest of the proof guarantees that this construction is arbitrarily close to the one where $\rho=n$ as $\rho\to n^+$. To do so, I will rely on the fact that the setup under consideration satisfies the assumptions for Theorem 7.1.1 (existence and uniqueness) as well as 7.3.1 (differentiable dependence on initial conditions and parameters) in \cite{Wiggins2003introduction}.
\end{remark}

\noindent\textbf{Eigenspaces.}
Following \cite{Wiggins2003introduction}, Section 3.1, let $E^{s}(\rho)$ and $E^{u}(\rho)$ denote the
(one-dimensional) stable and unstable subspaces of the vector field linearized around its hyperbolic steady state, \textit{i.e} the eigenspaces of the Jacobian of $F$, $DF(x^{\ast},\rho)$
associated with $\lambda_{s}$ and $\lambda_{u}$ respectively:
\[
  E^{s}(\rho) = \operatorname{span}\{e^{s}(\rho)\},
  \qquad
  E^{u}(\rho) = \operatorname{span}\{e^{u}(\rho)\},
\]
where $e^{s}(\rho)$ and $e^{u}(\rho)$ are eigenvectors corresponding to
$\lambda_{s}(\rho)$ and $\lambda_{u}(\rho)$.  These subspaces depend
smoothly on $\rho$. Note that since both eigenvalues have nonzero real
part, the centre subspace is trivial: $E^{c}(\rho) = \{0\}$.
 
\begin{lemma}[Smooth dependence on $\rho$]
  \label{lem:smooth}
  Let $\Omega := \{(k,\rho) :
  k \in (\bar{k}(\rho), k^{*}(\rho)), \rho\geq n\}$.
  The saving function $s(k;\rho)$, extended to $\rho = n$ is $C^{1}$ on the interior of $\Omega$.
\end{lemma}
 
\begin{proof}
The proof consists in five steps.
 
\noindent\textbf{Step 1.  Local stable manifold in translated
coordinates.}
Since the fixed point $x^{\ast}(\rho)$ moves with $\rho$, I rely on the following change of variables
\begin{equation}\label{eq:translate}
  y := x - x^{\ast}(\rho)
  = \begin{pmatrix} k - k^{*}(\rho) \\ s \end{pmatrix},
\end{equation}
so that the fixed point sits at the origin $y = 0$ for every $\rho$. Accordingly, define the translated vector field
\[
  G(y;\rho) := F\bigl(y + x^{\ast}(\rho); \rho\bigr).
\]
Then $G(0;\rho) = 0$ for all $\rho$, and
$D_{y}G(0;\rho) = \mathbf{J}$,
which has a hyperbolic saddle at the origin with eigenvalues
$\lambda_{s}(\rho) < 0 < \lambda_{u}(\rho)$.
Since $F \in C^{\infty}$ and $x^{\ast} \in C^{\infty}$, the composition
$G$ is $C^{p}$ jointly in $(y,\rho)$ for every $p \geq 1$. Now define $z := (y,\rho) \in \mathbb{R}^{3}$ and the augmented
autonomous system
\begin{equation}\label{eq:augmented}
  \dot{z} =
  \begin{pmatrix} G(y;\rho) \\ 0 \end{pmatrix}
\end{equation}
The fixed-point set of this vector field is the line
$\mathcal{L} := \{(0_{2\times 1},\rho) : \rho \in \mathcal{V}\}$, where
$\mathcal{V}$ is an open interval containing $n$. For each
fixed $\rho$, the linearization of this vector field at $(0,\rho)$ has eigenvalues
$\lambda_{s}(\rho)$, $\lambda_{u}(\rho)$, and $0$.  In the notation of
\cite{Wiggins2003introduction} (Section 3.1), the stable, unstable, and centre subspaces of
the augmented system are then:
\[
  \widetilde{E}^{s}(\rho) = \operatorname{span}\left\{(e^s(\rho),0)\right\},
  \qquad
  \widetilde{E}^{u}(\rho) = \operatorname{span}\left\{(e^u(\rho),0)\right\},
  \qquad
  \widetilde{E}^{c}(\rho) = \operatorname{span}\left\{(0,0,1)\right\},
\]
where $e^s(\rho),e^u(\rho)\in\mathbb{R}^2$ are the stable and unstable eigenvectors.
 
I now apply \textbf{Theorem~3.2.1} in \cite{Wiggins2003introduction} to the augmented system
\eqref{eq:augmented} at the fixed point $(0,\rho_{0})$ for any $\rho_{0}$
near $n$.  The theorem guarantees the existence of a $C^{p}$
one-dimensional local stable manifold $W^{s}_{\mathrm{loc}}(0;\rho_{0})$,
a $C^{p}$ one-dimensional local unstable manifold
$W^{u}_{\mathrm{loc}}(0;\rho_{0})$, and a $C^{p}$ one-dimensional local
centre manifold $W^{c}_{\mathrm{loc}}(0;\rho_{0})$ in $\mathbb{R}^{3}$, all
tangent to the respective subspaces at the origin.
 
What is needed in the current context is the \emph{joint} smooth dependence of $W^{s}_{\mathrm{loc}}$ on the parameter $\rho$ viewed as a coordinate in
the augmented space.  This is provided by the centre-stable manifold
structure.  In the augmented system, the centre-stable subspace is
$\widetilde{E}^{cs} = \widetilde{E}^{c} \oplus \widetilde{E}^{s}$, which
is two-dimensional.  By the same Theorem~3.2.1 applied to the splitting
$\mathbb{R}^{3} = \widetilde{E}^{cs} \oplus \widetilde{E}^{u}$ (treating the
unstable direction as the complementary subspace), the local
centre-stable manifold $W^{cs}_{\mathrm{loc}}$ is a $C^{p}$
two-dimensional surface in $\mathbb{R}^{3}$ tangent to
$\widetilde{E}^{cs}$ at $(0,\rho_{0})$.  Its intersection with each
constant-$\rho$ slice is precisely the local stable manifold at that
parameter value.  The $C^{p}$ regularity of $W^{cs}_{\mathrm{loc}}$ in the
augmented coordinates $(y_{1}, y_{2}, \rho)$ therefore yields a $C^{p}$
function
$\varphi\colon (-\eta,\eta) \times \mathcal{V}' \to \mathbb{R}$
(with $\eta > 0$ and $\mathcal{V}'$ an open neighborhood of $n$) that links $y_2$ (savings) to both $y_1$ (capital, in deviation from steady state) as well as the discount factor $\rho$:
\begin{equation}\label{eq:localmanifold}
  \bigl\{(y_{1},\varphi(y_{1};\rho))
    : y_{1} \in (-\eta,\eta)\bigr\}.
\end{equation}
Note that, by considering an open neighborhood of $n$ one implicitly considers values such that $\rho<n$. These are only considered to apply standard Theorems in the dynamical systems literature. Later on, the main result will be restricted to the domain where $\rho>n$. Translating back to the original coordinates via~\eqref{eq:translate},
the local saddle path is the graph
\begin{equation}\label{eq:sigmaloc}
  s = \sigma_{\mathrm{loc}}(k;\rho)
    := \varphi\bigl(k - k^{*}(\rho);\rho\bigr),
\end{equation}
defined for $|k - k^{*}(\rho)| < \eta$.  Since $k^{*}$ is $C^{\infty}$
in $\rho$ and $\varphi$ is $C^{p}$ in its arguments,
$\sigma_{\mathrm{loc}}$ is $C^{p}$ jointly in $(k,\rho)$.
 
\noindent\textbf{Step~2.  Anchor point on the local manifold.}
Fix a small constant $\eta_{0} \in (0,\eta)$ and define the anchor
point
\[
  x_{0}(\rho) :=
  \bigl(k^{*}(\rho) - \eta_{0},
  \sigma_{\mathrm{loc}}(k^{*}(\rho) - \eta_{0};\rho)\bigr).
\]
By construction $x_{0}(\rho) \in W^{s}_{\mathrm{loc}}(0;\rho)$ for
every $\rho \in \mathcal{V}'$, and $\rho \mapsto x_{0}(\rho)$ is
$C^{1}$ on $\mathcal{V}'$.  At $\rho = n$, the second component equals
$s(k^{*}(n) - \eta_{0};n) > 0$, since the closed-form saddle path is
strictly positive below $k^{*}(n)$.

The lower bound that is needed below, however, is not merely at the anchor but
\emph{uniform along the whole backward orbit}, and it is obtained as follows. The horizon is
fixed first, from the reference orbit alone, so that no circularity arises between the choice of
horizon and the choice of $\epsilon_{0}$. Consider the reference parameter $\rho=n$. The map
$t\mapsto K(t;n)$ is continuous and strictly decreasing, so there exists a finite $T^{\ast}>0$
with
\[
  K(T^{\ast};n)\;<\;k_{1},
\]
that is, the reference backward orbit has strictly passed $k_{1}$ by time $T^{\ast}$. Fix such a
$T^{\ast}$; it depends only on the reference orbit and on $k_{1}$, not on $\rho$.

\emph{Lower bound on the reference orbit.} On the compact segment $[0,T^{\ast}]$ the saving
component $t\mapsto S(t;n)$ is continuous and strictly positive: the orbit remains in the region
$k\ge K(T^{\ast};n)>0$, where Lemma~\ref{lem:no-obstruction}\emph{(iii)} gives
$s\in(0,1/\gamma)$. A positive continuous function on a compact set attains its minimum, so
\[
  2\underline{s} \;\defeq\; \min_{t\in[0,T^{\ast}]} S(t;n) \;>\; 0 .
\]

\emph{Transfer to nearby $\rho$.} By Theorem 7.3.1 in \cite{Wiggins2003introduction} the flow
depends continuously on the parameter, uniformly on compact time intervals. Applying this on the
\emph{fixed} horizon $[0,T^{\ast}]$, one may choose $\epsilon_{0}>0$ small enough that, for every
$\rho\in[n,n+\epsilon_{0}]$, both
\[
  \sup_{t\in[0,T^{\ast}]}\bigl|S(t;\rho)-S(t;n)\bigr|<\underline{s}
  \qquad\text{and}\qquad
  \sup_{t\in[0,T^{\ast}]}\bigl|K(t;\rho)-K(t;n)\bigr|<\tfrac12\bigl(k_{1}-K(T^{\ast};n)\bigr).
\]
The first inequality gives, for all $t\in[0,T^{\ast}]$ and all $\rho\in[n,n+\epsilon_{0}]$,
\[
  S(t;\rho)\;\ge\;S(t;n)-\underline{s}\;\ge\;2\underline{s}-\underline{s}\;=\;\underline{s}\;>\;0,
\]
the second inequality in this chain holding because the minimum defining $2\underline{s}$ was
taken over the same interval $[0,T^{\ast}]$. The second inequality above gives
$K(T^{\ast};\rho)<k_{1}$, so the perturbed orbit also reaches $k_{1}$ before time $T^{\ast}$; the
hitting time therefore satisfies $T_{1}(\rho)<T^{\ast}$ for every admissible $\rho$, and
$T^{\ast}$ serves as a uniform horizon.

The bound $\underline{s}$ is therefore a single positive constant valid along the entire
backward orbit, on the common horizon $[0,T^{\ast}]$ and for every admissible $\rho$, which is
what is used in Steps 4 and 5. Note that only positivity on a compact segment of the
\emph{reference} orbit and finite-time continuous dependence are required; no monotonicity of
$s$ in $k$ is invoked.
  
\noindent\textbf{Step 3.  Backward flow from the anchor.}
As before, let $x(t,x_0;\rho)$ denote the flow of \eqref{eq:vector_field},
i.e.\ the solution at time~$t$ starting from $x_{0}$ at $t = 0$
under parameter~$\rho$. From Theorem 7.3.1 in \cite{Wiggins2003introduction},
$x(t,x_0;\rho)$ is $C^{p}$ jointly in $(t, x_{0}, \rho)$ on its maximal
domain of definition (with continuation guaranteed by
Theorem 7.2.1). Consider the backward orbit starting from the anchor:
\begin{equation}\label{eq:backward}
  \beta(t;\rho) := x(-t,x_0;\rho),
  \qquad t \geq 0.
\end{equation}
Since $x(t,x_0;\rho)$ is $C^{p}$ in $(t,x_{0},\rho)$ and
$x_{0}$ is $C^{1}$ in~$\rho$, the composition $\beta$ is $C^{1}$
jointly in $(t,\rho)$. Now write $\beta(t;\rho) \defeq \bigl(K(t;\rho), S(t;\rho)\bigr)$, the components of the \emph{backward} orbit, so that $\partial K/\partial t=-S$, the minus sign coming from the time reversal and not from any capital-letter convention.
Along the saddle path the $s$-component is strictly positive
(away from the equilibrium), so $\dot{k} = s > 0$ in forward time;
equivalently $\partial K / \partial t = -S < 0$ in the backward
direction.  Hence $K(t;\rho)$ is strictly decreasing in~$t$,
with $K(0;\rho) = k^{*}(\rho) - \eta_{0}$.
 
\noindent\textbf{Step 4. Uniform backward time to cover
$[k_{1},k_{2}]$.}
Let $[k_{1},k_{2}]$ be a closed sub-interval of
$(\bar{k}(n),k^{*}(n))$.  Since $k_{2} < k^{*}(n)$ strictly and
$k^{*}(\rho) \to k^{*}(n)$ as $\rho \to n$, one may shrink
$\mathcal{V}'$ to ensure $k_{2} < k^{*}(\rho) - \eta_{0}$ for all
$\rho \in \mathcal{V}'$: the anchor lies strictly to the right of
$[k_{1},k_{2}]$ uniformly. Define the hitting time
\[
  T_{1}(\rho) := \inf\bigl\{t > 0 : K(t;\rho) = k_{1}\bigr\},
\]
which is finite because $K(t;\rho)$ is strictly decreasing.
Since $\partial K/\partial t = -S \leq -\underline{s} < 0$ on $[0,T^{\ast}]$ by Step~2, the
implicit function theorem applies and yields $T_{1} \in C^{1}$ near $\rho=n$. No supremum over
$\rho$ is needed here: Step~2 already established that $T_{1}(\rho) < T^{\ast}$ for every
$\rho\in[n,n+\epsilon_{0}]$, with $T^{\ast}$ fixed in advance from the reference orbit alone.
The horizon $T^{\ast}$ is therefore uniform in $\rho$ by construction rather than by a
compactness argument applied after the fact. On the compact set
$[0,T^{\ast}] \times [n,n+\epsilon_0]$, the map $\beta$ is uniformly $C^{1}$. Remember that time
is flowing backwards here.
 
\noindent\textbf{Step 5. Global inversion to extract $s(k;\rho)$.}
As in step 4, fix the compact parameter interval $[n, n+\epsilon_{0}] \subset \mathcal{V}'$ from Step~4. On $[0,T^{\ast}] \times [n, n+\epsilon_{0}]$ define the map
\[
  \Gamma \colon [0,T^{\ast}] \times [n, n+\epsilon_{0}] \to \mathbb{R} \times [n, n+\epsilon_{0}],
  \qquad
  \Gamma(t,\rho) := \bigl(K(t;\rho),\rho\bigr),
\]
which is $C^{1}$ by Step~3, with Jacobian
\[
  D\Gamma =
  \begin{pmatrix}
    \partial K/\partial t & \partial K/\partial\rho \\
    0 & 1
  \end{pmatrix},
  \qquad
  \det D\Gamma = \frac{\partial K}{\partial t} = -S \le -\underline{s} < 0 ,
\]
the bound $-S \le -\underline{s}$ holding uniformly by the argument of Step~2, which supplies a
single constant $\underline{s}>0$ along the whole backward orbit and for every
$\rho\in[n,n+\epsilon_{0}]$.

Since the flow is defined for times slightly beyond $T^{\ast}$ and for parameters in the open
neighbourhood $\mathcal{V}'$, the map $\beta$, and hence $\Gamma$, is defined and $C^{1}$ on an
open neighbourhood of the compact rectangle $[0,T^{\ast}]\times[n,n+\epsilon_{0}]$. Every point of
the rectangle is therefore an interior point of the domain of $\Gamma$, and the inverse function
theorem applies at each of them in its standard form.

It will be useful to show that $\Gamma$ is a $C^{1}$ diffeomorphism onto its image. \emph{Injectivity.} $\Gamma$ preserves the second coordinate, so it suffices to show that $t\mapsto K(t;\rho)$ is injective for each fixed $\rho$. Since $\partial K/\partial t = -S \le -\underline{s} < 0$ on all of $[0,T^{\ast}]$, the map $t\mapsto K(t;\rho)$ is strictly decreasing, hence injective; consequently $\Gamma$ is injective on $[0,T^{\ast}]\times[n, n+\epsilon_{0}]$. \emph{Local regularity.} Since $\det D\Gamma \ne 0$ everywhere, the inverse function theorem makes $\Gamma$ a local $C^{1}$ diffeomorphism at each point. A $C^{1}$ map that is both globally injective and everywhere a local diffeomorphism is a diffeomorphism onto its (open) image, with $C^{1}$ inverse
\[
  \Gamma^{-1}(k,\rho) = \bigl(t(k;\rho),\rho\bigr).
\]
Equivalently and more concretely, for each $\rho$ the strictly monotone $C^{1}$ function $K(\cdot;\rho)$ has a $C^{1}$ inverse $t(\,\cdot\,;\rho)$ on its range by the one-dimensional inverse function theorem, and $t$ depends $C^{1}$ on $\rho$ because $\Gamma$ is jointly $C^{1}$ with $\partial K/\partial t$ bounded away from zero.

\emph{Domain.} By Step~4 the anchor satisfies $K(0;\rho)=k^{*}(\rho)-\eta_{0} > k_{2}$ and the hitting time $T_{1}(\rho)$ with $K(T_{1}(\rho);\rho)=k_{1}$ is bounded by $T^{\ast}$, uniformly in $\rho\in[n, n+\epsilon_{0}]$; hence for every $\rho\in[n, n+\epsilon_{0}]$ the range of $K(\cdot;\rho)$ over $[0,T^{\ast}]$ contains $[k_{1},k_{2}]$, so $t(k;\rho)$ is defined on an open set containing the compact rectangle
$[k_{1},k_{2}]\times[n, n+\epsilon_{0}]$, and in particular on an open set containing
$[k_{1},k_{2}]\times\{n\}$. Finally define
\begin{equation}\label{eq:sfinal}
  s(k;\rho) := S\bigl(t(k;\rho);\rho\bigr).
\end{equation}
As a composition of $C^{1}$ maps, $s$ is $C^{1}$ jointly in $(k,\rho)$. By construction the point $\bigl(k,s(k;\rho)\bigr)$ lies on the saddle path for the parameter $\rho$. Since $[k_{1},k_{2}]$ was an arbitrary compact sub-interval of $(\bar{k}(n),k^{*}(n))$, the saving function is $C^{1}$ on the interior of~$\Omega$. 
Throughout the construction, the compact region
$\mathcal{K} := [k_{1}, k^{*}(n) - \eta_{0}/2] \times
[\underline{s}, 1/\gamma]$ in which the backward orbit lives
contains no fixed point of~\eqref{eq:vector_field} other than
$x^{\ast}(\rho)$, which lies outside $\mathcal{K}$ to the right. Indeed, for $\rho$ close enough to $n$ (or $\epsilon_0$ being small enough), this guarantees that the upper bound for $k$ doesn't depend on $\rho$ and is such that $k^{*}(n) - \eta_{0}/2<k^{*}(\rho)$.

By Lemma \ref{lem:no-obstruction}, the vector field is $C^{\infty}$ on~$\mathcal{K}$,
the orbit is bounded away from the singular boundary $\{k = 0\}$,
and there is no homoclinic connection.  The compactness and
non-degeneracy of $\mathcal{K}$ justify applying
Theorem 7.3.1 in \cite{Wiggins2003introduction} on the uniform time interval $[0,T^{\ast}]$.
\end{proof}
 
Before stating the main theorem, the following Lemma will be useful.

\begin{lemma}[Subsistence threshold] \label{lem:threshold}
Let $c_0(k):=f(k) - (\delta+n)k-s(k;n)$ be the unconstrained consumption function. There exists a unique $\bar{k}(n) \in (\kmin, k^{*}(n))$ with $c_{0}(\bar{k}(n)) = \bar{c}$; moreover $c_{0}(k) < \bar{c}$ for $k \in [\kmin,\bar{k}(n))$ and $c_{0}(k) > \bar{c}$ for $k \in (\bar{k}(n), k^{*}(n))$.
\end{lemma}
\begin{proof}
On $(\kmin, k^{\ast}(n))$ the closed form of Proposition \ref{prop:closed_form} gives $s(k;n)$ explicitly, so $c_0=\psi-s(\cdot\,;n)$ is well defined and continuously differentiable there. Differentiating and using the saddle-path ODE \eqref{eq:ODE_rho_n} at $\rho=n$ to substitute for $s'(k;n)$,
\[
c_0'(k)=\psi'(k)-s'(k;n)
=\psi'(k)-\frac{\psi'(k)}{s(k;n)}\left[s(k;n)-\frac1\gamma\right]
=\frac{\psi'(k)}{\gamma\, s(k;n)}
=\frac{r(k)-n}{\gamma\, s(k;n)},
\]
which is strictly positive because $\psi'(k)=r(k)-n>0$ for $k<k^{\ast}(n)$ and $s(k;n)\in(0,1/\gamma)$ by Lemma \ref{lem:no-obstruction}\emph{(iii)}. Hence $c_0$ is continuous and strictly increasing on $(\kmin,k^{\ast}(n))$.

At the left endpoint, $\psi(\kmin)=\cbar$ by definition of $\kmin$, so $c_0(\kmin)=\psi(\kmin)-s(\kmin;n)=\cbar-s(\kmin;n)<\cbar$ strictly, since $s>0$ on the saddle path (Lemma \ref{lem:no-obstruction}(iii)). At the right endpoint $c_0(k^\ast(n))=c^\ast>\cbar$ by Assumption \ref{ass:k_min_c_bar}. The intermediate value theorem gives a unique $\bar k(n)\in(\kmin,k^\ast(n))$ with $c_0(\bar k(n))=\cbar$, and strict monotonicity gives the sign pattern; in particular $\bar k(n)>\kmin$ is a conclusion, not an assumption.
\end{proof}

\subsection{Main Theorem Statement}

\setcounter{theorem}{0}
\begin{theorem}[Uniform convergence on compacta] \label{thm:main}
Let $[k_{1}, k_{2}]$ be a closed interval with $\kmin \le k_{1} < k_{2} < k^{*}(n)$ and
$\bar{k}(n) \notin [k_{1}, k_{2}]$, and let the limiting policy be given piecewise by
\eqref{eq:piecewise_limit}. There exist $\epsilon > 0$ and a constant $C
\geq 0$ (depending on $[k_{1}, k_{2}]$ and on the parameters $\gamma, \alpha, \delta, n$)
such that, for every $\rho$ satisfying $0 < \rho - n < \epsilon$,
\begin{equation} \label{eq:bound_app}
\sup_{k \in [k_{1}, k_{2}]} \bigl| s(k; \rho) - s(k;n) \bigr| \leq
C (\rho - n),
\end{equation}
and an analogous bound holds for $c(k;\rho) - c(k;n)$.
\end{theorem}

\begin{proof}
It has been established in section \ref{sec:model} that consumption is given by $\cbar$ if $k\leq \bar{k}(\rho)$ and that optimal consumption follows the Euler equation if $k>\bar{k}(\rho)$. The threshold $\bar{k}(\rho)$ is thus defined by $c(\bar{k}(\rho); \rho) = \cbar$. By construction, optimal consumption is
jointly continuous in $(k, \rho)$ on the interior. Given that the Euler equation holds for $k\in(\bar{k}(\rho),k^\ast(\rho))$, one can write
\begin{align*}
c'(k) = \frac{\dotc}{\dotk} = \frac{r(k)-\rho}{\gamma s(k)}>0    
\end{align*}
so that consumption is strictly increasing in capital. In particular, one has $c'(\bar{k}(n)) > 0$ so that the implicit
function theorem applies and yields $\bar{k}(\rho) \in C^{1}$ near $\rho = n$. In
particular, $\bar{k}(\rho) \to \bar{k}(n)$ as $\rho \to n^+$. Given the maintained hypothesis that
$\bar{k}(n) \notin [k_{1}, k_{2}]$, exactly one of the following two cases obtains.

\emph{Case 1: $[k_{1}, k_{2}] \subset [\kmin, \bar{k}(n))$.} Set $\eta := \bar{k}(n) -
k_{2} > 0$. By continuity of $\bar{k}(\rho)$, there exists $\epsilon > 0$ such that
$\bar{k}(\rho) > \bar{k}(n) - \eta / 2 > k_{2}$ for all $\rho \in (n, n +
\epsilon)$. For such $\rho$, both the unperturbed and perturbed problems are in
the constrained regime throughout $[k_{1}, k_{2}]$: $c(k; \rho) = \bar{c}$
for all $k \in [k_{1}, k_{2}]$. Hence \eqref{eq:bound} holds with $C = 0$.

\medskip
\emph{Case 2: $[k_{1}, k_{2}] \subset (\bar{k}(n), k^{*}(n))$.} Set $\eta_{1} := k_{1}
- \bar{k}(n) > 0$ and $\eta_{2} := k^{*}(n) - k_{2} > 0$. By continuity of $\bar{k}(\rho)$
and $k^{*}(\rho)$, there exists $\epsilon > 0$ such that for all $\rho \in [n, n +
\epsilon]$,
\begin{equation*}
\bar{k}(\rho) < k_{1} - \eta_{1}/2 \quad \text{and} \quad k^{*}(\rho) > k_{2} +
\eta_{2}/2.
\end{equation*}
For such $\rho$, the compact $[k_{1}, k_{2}]$ lies in the unconstrained regime of both
problems, and $c(k;\rho) > \bar{c}$ on $[k_{1}, k_{2}]$. By
Lemma \ref{lem:smooth}, the map $(k, \rho) \mapsto s(k; \rho)$ is $C^{1}$ on the compact
rectangle $[k_{1}, k_{2}] \times [n, n + \epsilon]$. Its partial derivative
$\partial s/\partial \rho$ is therefore continuous on this compact set and attains its
supremum
\begin{equation*}
M := \sup_{(k, \rho) \in [k_{1}, k_{2}] \times [n, n + \epsilon]} \bigl|
\partial s(k; \rho)/\partial \rho \bigr| < \infty.
\end{equation*}
Let $s_{\rho}(k;\rho)\defeq \p s(k;\rho)/\p \rho$. Then, by the mean value theorem, for each fixed $k \in [k_{1}, k_{2}]$ and each $\rho \in (n,
n + \epsilon)$, there exists $\rho' \in (n, \rho)$ such that
\begin{equation*}
s(k; \rho) - s(k; n) = (\rho - n)
s_\rho(k; \rho').
\end{equation*}
Taking absolute values, $|s(k; \rho) - s(k;n)| \leq M (\rho - n)$. Taking the supremum
over $k \in [k_{1}, k_{2}]$,
\begin{equation*}
\sup_{k \in [k_{1}, k_{2}]} |s(k; \rho) - s(k;n)| \leq M (\rho - n).
\end{equation*}
Since $c(k; \rho) = f(k) - (\delta+n)k - s(k; \rho)$ and $c(k;n) = \Psi_{n}(k) - s(k;n)$,
and $\Psi$ is independent of $\rho$, subtracting gives $c(k; \rho) - c(k;n) =
-(s(k; \rho) - s(k;n))$, hence the same bound applies to $c$. Setting $C=\max\left\{0,M\right\}$ completes
the proof.
\end{proof}

\section{Proof of Proposition \ref{prop:threshold_gamma}}
\label{sec:app_proof_threshold_gamma}

\noindent\textbf{Step 1: speed at the steady state.} Since $s(k^\ast)=0$, L'H\^opital's rule gives $\beta(k^\ast;\gamma)=-s'(k^\ast)$. Linearising $\dotk=s(k)$ around $k^\ast$ gives $s'(k^\ast)=\lambda_s$, and with $\rho=n$ the characteristic equation $\lambda^2-(\rho-n)\lambda+f''(k^\ast)/\gamma=0$ reduces to $\lambda=\pm\sqrt{|f''(k^\ast)|/\gamma}$. Hence $\beta(k^\ast;\gamma)=\sqrt{|f''(k^\ast)|/\gamma}$.

\noindent\textbf{Step 2: the junction in closed form.} At the take-off point consumption equals
subsistence, so $s(\bar{k}(n);n)=\psi(\bar{k}(n))-\cbar$. Writing
$\bar q\defeq\gamma s(\bar{k}(n);n)=\gamma(\psi(\bar{k}(n))-\cbar)$ and evaluating
\eqref{eq:gamma_sk_Lambert} at $k=\bar{k}(n)$,
\[
\bar q+\ln(1-\bar q)=\gamma\bigl(\psi(\bar{k}(n))-\psi(k^{\ast}(n))\bigr)
=\bar q-\gamma\Delta^{\ast},
\]
where $\Delta^{\ast}=\psi(k^{\ast}(n))-\cbar$. The terms in $\bar q$ cancel, leaving
$\ln(1-\bar q)=-\gamma\Delta^{\ast}$, that is $\bar q=1-e^{-\gamma\Delta^{\ast}}$. This is
\eqref{eq:junction_closed_form}; note that it is elementary and involves no Lambert function.
Consequently
\begin{equation}\label{eq:R_junction}
R(\gamma)=\frac{1-e^{-\gamma\Delta^{\ast}}}{\sqrt{\gamma}\,\bigl(k^{\ast}(n)-\bar{k}(n)\bigr)\sqrt{|\psi''(k^{\ast}(n))|}},
\end{equation}
where $\bar{k}(n)$ is determined implicitly by $\psi(\bar{k}(n))=\cbar+(1-e^{-\gamma\Delta^{\ast}})/\gamma$.

\noindent\textbf{Step 3: the limit as $\gamma\to\infty$.} As $\gamma\to\infty$, the fact that $k^{\ast}(n)$ is independent from $\gamma$ implies that
$(1-e^{-\gamma\Delta^{\ast}})/\gamma\to0$, so $\psi(\bar{k}(n))\to\cbar$ and hence
$\bar{k}(n)\to\kmin$. The numerator of \eqref{eq:R_junction} tends to $1$ while the denominator
grows like $\sqrt{\gamma}\,(k^{\ast}(n)-\kmin)\sqrt{|\psi''(k^{\ast}(n))|}$, so
$R(\gamma)\to0$. Economically, high risk aversion drives the take-off point down towards the
subsistence floor and the saving rate there to zero, so convergence is slow at the take-off
relative to the steady state.

\noindent\textbf{Step 4: the limit as $\gamma\to0^{+}$.} As $\gamma\to0^{+}$, expanding the exponential gives
$(1-e^{-\gamma\Delta^{\ast}})/\gamma\to\Delta^{\ast}$, so $\psi(\bar{k}(n))\to\cbar+\Delta^{\ast}=\psi(k^{\ast}(n))$
and therefore $\bar{k}(n)\to k^{\ast}(n)$: the take-off point migrates to the steady state.
Both the numerator and the factor $k^{\ast}(n)-\bar{k}(n)$ in \eqref{eq:R_junction} vanish, so the
limit must be computed to second order. Write $x\defeq k^{\ast}(n)-\bar{k}(n)$,
$A\defeq|\psi''(k^{\ast}(n))|=-\psi''(k^{\ast}(n))>0$ given that $\psi''<0$ and $B\defeq\psi'''(k^{\ast}(n))$. Given that we don't have $\bar{k}(n)$ in closed form, it will be useful to work with $\psi(\bar{k}(n))$ for which we do. Up to a third order, one can write:
\begin{equation}\label{eq:P_expansion}
P(\gamma)\defeq\psi(k^{\ast}(n))-\psi(\bar{k}(n))
=\Delta^{\ast}-\frac{1-e^{-\gamma\Delta^{\ast}}}{\gamma}
=\frac{\gamma(\Delta^{\ast})^{2}}{2}-\frac{\gamma^{2}(\Delta^{\ast})^{3}}{6}+O(\gamma^{3}).
\end{equation}
On the other hand, since $\psi'(k^{\ast}(n))=0$, taking a third-order Taylor expansion of $\psi$ about
$k^{\ast}(n)$, evaluating it at $\bar{k}(n)$ and re-arranging gives:
\begin{equation}\label{eq:P_taylor}
P(\gamma)=\frac{A}{2}x^{2}+\frac{B}{6}x^{3}+O(x^{4}).
\end{equation}
Matching \eqref{eq:P_expansion} and \eqref{eq:P_taylor} at leading order yields
$x=\Delta^{\ast}\sqrt{\gamma/A}\,(1+\varepsilon)$ with $\varepsilon\to0$, which implies the ansatz $x\simeq \sqrt{\gamma}$. Substituting this into equation \eqref{eq:P_taylor} and using the fact that $\varepsilon^2$ is negligible against $\varepsilon$ as $\varepsilon\to 0$, one can write
\begin{align}
\label{eq:def_varepsilon}
\varepsilon=-\frac{B\,\Delta^{\ast}}{6A^{3/2}}\sqrt{\gamma}+O(\gamma).
\end{align}
Inserting $x$ and the numerator expansion $1-e^{-\gamma\Delta^{\ast}}=\gamma\Delta^{\ast}(1-\gamma\Delta^{\ast}/2+O(\gamma^{2}))$
into \eqref{eq:R_junction}, the factors $\gamma\Delta^{\ast}$ cancel and
\[
R(\gamma)=\frac{1-\gamma\Delta^{\ast}/2+O(\gamma^{2})}{1+\varepsilon}
\]
Now using the fact that $1/(1+\varepsilon)=1-\varepsilon+O(\varepsilon^2)$ alongside equation \eqref{eq:def_varepsilon}, one can finally write that
\begin{align*}
R(\gamma)=1+\frac{B\,\Delta^{\ast}}{6A^{3/2}}\sqrt{\gamma}+O(\gamma).    
\end{align*}
For a Cobb--Douglas production function $\psi'''(k)=f'''(k)=\alpha(\alpha-1)(\alpha-2)k^{\alpha-3}>0$,
so $B>0$ and the coefficient of $\sqrt{\gamma}$ is strictly positive. Hence
$R(\gamma)\to1$ from above as $\gamma\to0^{+}$, which is claim \emph{(i)}. 

\noindent\textbf{Step 5: conclusion.} By Step 4, $R(\gamma)>1$ for all sufficiently small
$\gamma>0$, and by Step 3, $R(\gamma)<1$ for all sufficiently large $\gamma$. Since $R$ is
continuous on $(0,\infty)$, the intermediate value theorem yields a threshold
$\underline\gamma>0$ with $R(\underline\gamma)=1$, front-loading below it and back-loading above
it.

\end{document}